\documentclass[letterpaper, 10 pt, conference]{ieeeconf}  

\IEEEoverridecommandlockouts                              

\usepackage[normalem]{ulem}

\usepackage{amsmath} 
\usepackage{amssymb}  
\usepackage{url}
\usepackage{hyperref}
\usepackage{soul,color}
\usepackage{threeparttable}
\usepackage{dcolumn}
\usepackage{mathtools}
\usepackage{subfig}
\usepackage{comment}
\usepackage{graphicx}
\usepackage{epstopdf}
\usepackage{cite}
\usepackage{amsfonts}
\usepackage{psfig}
\usepackage{amsmath}
\usepackage{subfig}
\usepackage{graphicx}
\usepackage{rotating}
\usepackage{enumerate}
\usepackage{epigraph}
\usepackage{float}
\usepackage{psfig}
\usepackage{epsfig}
\usepackage{graphicx}
\usepackage{amsfonts}
\usepackage{amsmath}
\usepackage{amssymb}
\usepackage{color}
\usepackage{pstricks}
\usepackage{graphics}
\usepackage{rotating}
\usepackage{lettrine}
\usepackage{multicol}

\newtheorem{lemma}{Lemma}
\newtheorem{theorem}{Theorem}
\newtheorem{cor}{Corollary}
\newcolumntype{d}{D{.}{.}{-1}}
\title{\LARGE \bf
Consensus over Weighted Directed Graphs: A Robustness Perspective
}

\author{Dwaipayan Mukherjee and Daniel Zelazo
\thanks{This work was supported in part at the Technion by a fellowship of the Israel Council for Higher Education and the Israel Science Foundation (grant No. 1490/1).
D. Mukherjee ({\tt\small dwaipayan.mukherjee2@gmail.com}) and D. Zelazo ({\tt\small dzelazo@technion.ac.il}) are with the Faculty of Aerospace Engineering, Technion - Israel Institute of Technology, Haifa 32000, Israel.} 
}

\begin{document}

\maketitle
\thispagestyle{empty}
\pagestyle{empty}

\begin{abstract}

The present paper investigates the robustness of the consensus protocol over weighted directed graphs using the Nyquist criterion. The limit to which a single weight can vary, while consensus among the agents can be achieved, is explicitly derived. It is shown that even with a negative weight on one of the edges, consensus may be achieved. The result obtained in this paper is applied to a directed acyclic graph and to the directed cycle graph. Graph theoretic interpretations of the limits are provided for the two cases. Simulations support the theoretical results. 

\end{abstract}

\section{Introduction}
\label{sec: intro}

\par{}The consensus protocol is an important problem in multi-agent systems, that has received a lot of attention \cite{beard}. In this context, some work on the robustness of undirected graphs has been carried out by merging concepts from graph theory and robust control \cite{DZ, DZ1}. These involve the application of the small gain theorem to the networked dynamic system described by the graph Laplacian and the edge Laplacian matrices. Particularly, \cite{DZ1} considered the possibility of admitting negative weights on some of the edges. The context in which negative edge weights arise are discussed therein and also in the special case of cyclic pursuit as in \cite{AS, DM1}. However, only undirected graphs, whose Laplacians are symmetric and therefore lend themselves to analysis, have been studied. This paper considers a weighted directed graph (digraph) for similar robustness studies. Thus the agents run a consensus protocol over a weighted digraph \cite{MesbahiEgerstedt}. 
It will be shown in this paper that even in the absence of symmetric Laplacians, robust stability analysis can be carried out for a special class of weighted digraphs. 

The networked system is first transformed to edge variables, leading to a directed \emph{edge agreement protocol}, originally studied in \cite{DZ} for undirected graphs. This work further develops properties of the \emph{directed edge Laplacian} matrix. Some recent work such as \cite{recent, DEL, DEL1} also present some results on the un-weighted edge Laplacian for a digraph. In \cite{DZ} the edge Laplacian aided in studying the roles of certain subgraphs such as cycles and spanning trees in the agreement problem. Both \cite{DZ} and \cite{DZ1} built the platform for robustness studies (performance and stability) of the consensus problem over undirected graphs. 
The main focus of this work is to consider the robust stability of the directed and weighted edge agreement protocol where uncertainty in the model is introduced in the form of a perturbation to one of the edge weights. 
The robust stability result for a general weighted digraph is first derived using the Nyquist criteria. Further analysis is then provided, along with graph-theoretic interpretations, for two specific classes of graphs - the directed acyclic graph and the directed cycle graph. 
It is shown that for a directed acyclic graph, robust stability requires the magnitude of the negative weight of the uncertain edge to be less than the sum of the nominal positive weights of its sibling edges. 
For the directed cycle graph, it is shown that the limit on the perturbation on a single edge weight is the same as the one obtained in the literature \cite{AS}, \cite{DM1}. In terms of graph resistance, this limit is such that the resistance of a perturbed edge, $e_k$, running from node $i$ to node $j$, must be at least equal to the negative of the equivalent graph resistance between nodes $i$ and $j$, with $e_k$ removed. 
\par{} Section \ref{sec:edgelap} describes the edge Laplacian for a weighted digraph and then some of its properties are stated. The robust stability of the uncertain edge protocol for a weighted digraph is analyzed in Section \ref{sec:robust}. 
Section~\ref{sec: sim} presents relevant simulations to support the results and Section~\ref{sec: conc} concludes the paper.
\paragraph*{Notation} The null space and range space of a matrix $A$ are denoted by $\mathcal{N}(A)$ and $\mathcal{R}(A)$, respectively. The vectors of all-ones and all-zeros in $\mathbb{R}^p$ are denoted by $\mathbf{1}_p$ and $\mathbf{0}_p$ respectively. 
A weighted digraph, $\mathcal{G}$, is specified by its vertex set $V$, the edge set $\mathcal{E}$ that captures the incidence relation between pairs of $V$, and the diagonal weight matrix $W$ which contains the weights of the edges. When the weights are all unity, the graph is represented by $V$ and $\mathcal{E}$ only. Throughout this paper, it is assumed that $|V|=n$ and $|\mathcal{E}|=m$. 

\section{Directed Weighted Edge Laplacian}\label{sec:edgelap}

The graph Laplacian matrix provides a beautiful link between discrete notions in graph theory to continuous representations, such as vector spaces and manifolds \cite{GR}.  Motivated by its role in consensus-seeking systems, an edge variant of the Laplacian, known as the edge Laplacian, was introduced in \cite{DZ}. In this section, an extension of this work is presented by considering directed and weighted graphs. As will be shown in Section \ref{sec:robust}, the edge Laplacian for digraphs provides the correct algebraic construction to analyze the robustness of consensus protocols over digraphs.

Some notions related to digraphs are first reviewed. A node $v \in V$ that can be reached by a directed path from every other node in $\mathcal{G}$ is termed a \emph{globally reachable node}. For any digraph containing at least one globally reachable node, a spanning subgraph $\mathcal{G}_\tau \subseteq \mathcal{G}$, termed a \emph{rooted in-branching}, is defined such that there exists a directed path from every node to a globally reachable node (or \emph{root}), and all other nodes, except this root with out-degree $0$, have out-degree equal to 1 in $\mathcal{G}_\tau$. For consensus over a digraph, there must be a globally reachable node, and hence a rooted in-branching \cite{AC}.
For a digraph with a rooted in-branching, another subgraph, $\mathcal{G}_c$, can be defined such that $\mathcal{G}_{\tau} \cup \mathcal{G}_{c} = \mathcal{G}$. The subgraph $\mathcal{G}_\tau$ has $n-1$ directed edges in the edge set $\mathcal{E}_\tau$, while the remaining $m-n+1$ edges constitute the edge set $\mathcal{E}_c$ corresponding to $\mathcal{G}_{c}$ (with $\mathcal{E}=\mathcal{E}_{\tau} \cup \mathcal{E}_c$). 

For undirected graphs, the graph and edge Laplacian matrices can be defined in terms of the incidence matrix, $E(\mathcal{G})$.  The incidence matrix is defined such that $[E(\mathcal{G})]_{ij}=1$ if edge $e_j$ is outgoing from vertex $i$, $[E(\mathcal{G})]_{ij}=-1$ if edge $e_j$ is incoming at vertex $i$, and $[E(\mathcal{G})]_{ij}=0$ otherwise. The graph Laplacian for a directed graph can be defined as $L_g=\mathcal{A}(\mathcal{G})E(\mathcal{G})^{T}$, where $\mathcal{A}(\mathcal{G}) \in \mathbb{R}^{n \times m}$ is such that $[\mathcal{A}(\mathcal{G})]_{ij}=1$ if the edge $e_j$ is outgoing from vertex $i$ and is 0 otherwise \cite{recent}. 
Similarly, $L_{e}=E(\mathcal{G})^{T}\mathcal{A}(\mathcal{G})$ is defined as the \emph{directed edge Laplacian}. 
The matrices $E(\mathcal{G})$ and $\mathcal{A}(\mathcal{G})$, for the digraph $\mathcal{G}$, may be written as $E$ and $\mathcal{A}$ for brevity. 

\par{}The graph Laplacian and the edge Laplacian for the weighted digraph $\mathcal{G}$ are given by $\bar{L}_g=\mathcal{A}(\mathcal{G})WE(\mathcal{G})^T$ and $\bar{L}_e=E(\mathcal{G})^{T}\mathcal{A}(\mathcal{G})W$, respectively, where, $W\in \mathbb{R}^{m\times m}$ is a diagonal matrix, whose diagonal entries are the weights of the corresponding edges, that is $W_{ii}=w_i>0~\forall i$.

\subsection{The Directed Edge Laplacian: Properties}
  
The directed edge Laplacian holds the key to the dynamics of the directed edge agreement problem. Hence, the important properties of $\bar{L}_e$ are central to an analysis of this problem. The following results, stated without proof, aid in that direction. Some recent works also focus on directed edge Laplacians with identical weights on all edges \cite{recent, DEL, DEL1}. For a nonsingular $W$, $\mathbf{dim}[\mathcal{N}(\mathcal{A})]=\mathbf{dim}[\mathcal{N}\mathbf(\mathcal{A}W)]$ and $\mathcal{R}(\mathcal{A})=\mathcal{R}(\mathcal{A}W)$. 
 
\begin{lemma}
For general weighted digraphs, $\mathcal{N}(\mathcal{A}W)\subseteq \mathcal{N}(\bar{L}_e)$. For weakly connected weighted digraphs, if there is at least one node with out-degree = 0, then $\mathcal{N}(\bar{L}_e)=\mathcal{N}(\mathcal{A}W)$, otherwise $\mathbf{1_n} \in \mathcal{R}(\mathcal{A})$ and $\mathcal{N}(\mathcal{A}W)\subset \mathcal{N}(\bar{L}_e)$.
\end{lemma}
\begin{proof}
Suppose $y$ belongs to $\mathcal{N}(\mathcal{A}W)$. Thus, $\mathcal{A}Wy=\mathbf{0}$. However, this also implies that $\bar{L}_{e}y=E^{T}\mathcal{A}Wy=\mathbf{0}$. 
Suppose $\mathbf{1_n} \in \mathcal{R}(\mathcal{A})$ implying every node has positive out-degree. 
For some nontrivial vector $x\in \mathbb{R}^{m}$, $\mathcal{A}x=\mathbf{1_n}$. Since $E^{T}\mathbf{1_n}=\mathbf{0}$, it follows that $\bar{L}_{e}(W^{-1}x)=E^{T}\mathcal{A}x=E^{T}\mathbf{1_n}=\mathbf{0}$ even though $W^{-1}x \not \in \mathcal{N}(\mathcal{A}W)$. This implies $\mathcal{N}(\mathcal{A}W)\subset \mathcal{N}(\bar{L}_e)$. 
Suppose $\mathbf{1_n}\not \in \mathcal{R}(\mathcal{A})$. Then there is at least one node whose out-degree is 0. If $\bar{L}_{e}x=\mathbf{0}$ for some $x\neq \mathbf{0}$, then either $\mathcal{A}Wx=\mathbf{0}$ or $\mathcal{A}Wx\in \mathcal{N}(E^T)$. Since for a weakly connected graph, $\mathbf{rank}(E^T)=n-1$ implying $\mathbf{dim}[\mathcal{N}(E^T)]=1$ and $\mathcal{N}(E^T)$ is spanned by $\mathbf{1_n}$, while $\mathcal{A}Wx \neq\mathbf{1_n}$ for any $Wx$, by assumption, it follows that $\mathcal{N}(\bar{L}_e)=\mathcal{N}(\mathcal{A}W)$.
\end{proof}
\begin{lemma}
The following statements are equivalent:
\begin{enumerate}
\item $\mathcal{A}$ has a nontrivial null space.
\item $\mathcal{A}$ has at least two identical columns.
\item The out-degree of at least one vertex in $\mathcal{G}$ is greater than unity. 
\end{enumerate}
\end{lemma}
\begin{proof}If $\mathcal{A}$ has a nontrivial null space, its columns, each corresponding to an edge of $\mathcal{G}$, are not linearly independent. Moreover, each column of $\mathcal{A}$ will have a $1$ at only one position and zeros elsewhere. Thus, for linear dependence, at least two columns of $\mathcal{A}$ must be identical and vice versa. But identical columns imply two edges emerging from the same vertex and vice versa. Thus, 1 $\Leftrightarrow$ 2 $\Leftrightarrow$ 3. 
\end{proof}
\begin{lemma}
If $\mathcal{G}$ has $r$ such vertices whose out-degrees are greater than or equal to 1, then $\mathbf{dim}[\mathcal{N}(\bar{L}_e)]\geq m-r$.
\end{lemma}
\begin{proof}Suppose there are $r$ vertices from which there is at least one outgoing edge. Choose one edge each that emerges from each of these $r$ vertices. The columns of $\mathcal{A}$ (or $\mathcal{A}W$) corresponding to these $r$ edges will be linearly independent. Hence, $\mathbf{rank}(\mathcal{A}W)=r$. By rank-nullity theorem \cite{HJ}, $\mathbf{dim}[\mathcal{N(\mathcal{A}W)}]= m-r$. Now, by Lemma 1, $\mathcal{N}(\mathcal{A}W)\subseteq \mathcal{N}(\bar{L}_e)$. This implies that $m-r \leq \mathbf{dim}[\mathcal{N}(\bar{L}_e)]$.
\end{proof}

\begin{lemma} If a digraph $\mathcal{G}$ has multiple globally reachable nodes, then they form directed cycle(s) in $\mathcal{G}$ and $\mathbf{1_n}\in \mathcal{R}(\mathcal{A})$.
\end{lemma}
\begin{proof} Suppose the set $V_{GR}$ contains all globally reachable nodes in the digraph $\mathcal{G}$. Then, for any two vertices $v_i, v_j \in  V_{GR}$, there exist directed paths from $v_i$ to $v_j$ and from $v_j$ to $v_i$, thus forming a directed cycle. Also, any node that is part of either path is also globally reachable. 
If there are multiple globally reachable nodes, then every node, including the globally reachable nodes, must have an out-degree greater than unity. Thus, $\mathcal{R}(\mathcal{A})=\mathcal{R}(I_n)$ and $\mathcal{R}(\mathcal{A})$ spans $\mathbb{R}^n$. So, $\mathbf{1_n}\in \mathcal{R}(\mathcal{A})$.
\end{proof}
\par{}Any vertex with an out-degree greater than unity contributes to $\mathcal{N}(\bar{L}_e)$, by Lemmas 1-3. Moreover, from Lemma 1, if $\mathbf{1_n} \in \mathcal{R}(\mathcal{A})$, then $\mathcal{N}(\bar{L}_e)\neq \mathcal{N}(\mathcal{A}W)$. By Lemma 4, a digraph having multiple globally reachable nodes must have a directed cycle among the globally reachable nodes and so every node must have an out-degree greater than 1. Thus, from Lemmas 1 and 4, $\mathcal{N}(\bar{L}_e)\neq \mathcal{N}(\mathcal{A}W)$ for such graphs. 

\subsection{Laplacians of Weighted Digraphs: Factorisations}
\label{sec:fact}
To understand the graph theoretic relation between the edges in $\mathcal{G}_\tau$ and $\mathcal{G}_c$ and to characterize the latter in terms of the former, the incidence matrix $E(\mathcal{G})$ 
can be factorized in certain forms. These factorisations also aid in the subsequent analysis in Section \ref{sec:robust}.
\begin{figure}[t]
     \centering
      \includegraphics[scale=0.750]{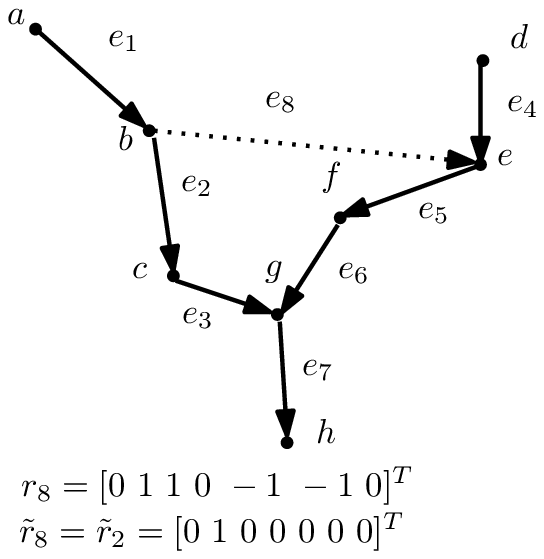}
      \caption{Dotted edge $e_8$ (sibling to edge $e_2$, with parent node $b$) encoded in terms of the edges in the rooted in-branching.}
      \label{demo}
      \vspace*{-0.5cm}
\end{figure}
  Define two edges outbound from the same node (parent node) as \emph{sibling edges}. Further, suppose that for the particular $\mathcal{G}_\tau$, the edges in $\mathcal{E}_\tau$ are labelled $e_1$ through $e_{n-1}$ with the corresponding parent nodes labelled $1$ through $n-1$. Clearly, no edge in $\mathcal{E}_\tau$ has a sibling in $\mathcal{G}_\tau$. The node with zero out-degree in $\mathcal{G}_\tau$ (which corresponds to any one globally reachable node, among possibly several, in $\mathcal{G}$) 
is labelled $n$. The incidence matrix is
\begin{align}
\label{inci}
\hspace{-5pt}E(\mathcal{G})=[E(\mathcal{G_\tau})~E(\mathcal{G}_c)]=E(\mathcal{G_\tau})[I_{n-1}~T_{\tau}]=E(\mathcal{G_\tau})R,
\end{align} 
where $T_{\tau} \in \mathbb{R}^{(n-1) \times (m-n+1)}$ may be given by
\begin{align}
\label{cite}
T_{\tau}=(E(\mathcal{G_\tau})^{T}E(\mathcal{G_\tau}))^{-1}E(\mathcal{G_\tau})^{T}E(\mathcal{G}_c),
\end{align}
as in \cite{DZ}. 
The matrix $E(\mathcal{G_\tau})^{T} \in \mathbb{R}^{(n-1)\times n}$ has full row rank and so the right inverse $E(\mathcal{G_\tau})(E(\mathcal{G_\tau})^{T}E(\mathcal{G_\tau}))^{-1}$ exists. Similarly, for a digraph with a single globally reachable node 
\begin{align}
\label{inc}
\hspace{-5pt}\mathcal{A}(\mathcal{G})=[\mathcal{A}(\mathcal{G_\tau})~\mathcal{A}(\mathcal{G}_c)]=\mathcal{A}(\mathcal{G_\tau})[I_{n-1}~\tilde{T}_{\tau}]=\mathcal{A}(\mathcal{G_\tau})\tilde{R},
\end{align} 
where $\tilde{T}_{\tau} \in \mathbb{R}^{(n-1) \times (m-n+1)}$, given by 
\begin{align}
\label{inci1}
 \tilde{T}_{\tau}=(A(\mathcal{G_\tau})^{T}A(\mathcal{G_\tau}))^{-1}A(\mathcal{G_\tau})^{T}A(\mathcal{G}_c),
\end{align}
encodes the siblings of edges in $\mathcal{E}_\tau$, that are in $\mathcal{E}_{c}$, while $\mathcal{A}(\mathcal{G_\tau})$ corresponds to edges in $\mathcal{E_\tau}$. 
For ${R}$, the last $m-n+1$ columns 
represent how the $m-n+1$ edges in $\mathcal{E}_{c}$ can be encoded in terms of the edges in $\mathcal{E}_\tau$ by a signed path vector \cite{DZ}, as illustrated in the example of Fig. \ref{demo}. A \emph{signed path} corresponding to an edge $e_i \in \mathcal{E}_{c}$ between nodes $a$ and $b$ in $\mathcal{G}$ is a sequence of edges in $\mathcal{G}_\tau$ 
such that this unoriented path leads from node $a$ to node $b$. Denote the $i$-th columns of $\tilde{R}$ and ${R}$ as $\tilde{r}_i$ and $r_{i}$, respectively, with $r_{i}(k)$ denoting the $k$-th entry of the column vector $r_{i}$. If the signed path corresponding to any of the edges $e_i \in \mathcal{E}_{c}$ involves traversing an edge $e_j \in \mathcal{G}_\tau$ in the same direction as its indicated direction in $\mathcal{G}_\tau$ (or $\mathcal{G}$), then $r_{i}(j)=+1$, whereas if it is traversed in a direction opposite to that marked on it, then the same entry is $-1$. If the signed path does not involve traversal of $e_j \in \mathcal{E}_\tau$, then $r_{i}(j)=0$. In the example of Fig. \ref{demo}, $e_8 \in \mathcal{E}_{c}$ is encoded in terms of $e_2$, $e_3$, $e_6$, and $e_5$ in $\mathcal{E}_\tau$. Thus, the corresponding entries in $r_8 \in \mathbb{R}^7$, are non-zero with the sign indicating the direction in which these edges are traversed (whether in the same direction as indicated by the arrowheads of the digraph, or opposite to it), 
while the other entries are zero. Also, every edge in $\mathcal{E}_{c}$ is a sibling edge to an edge in $\mathcal{E}_\tau$. So, the column in $\tilde{R}$ corresponding to any edge $e_q \in \mathcal{E}_{c}$, will be a replica of the column corresponding to its sibling edge in $\mathcal{E}_\tau$. 
Hence,
\begin{align}
\label{inter}
\tilde{r}_i= \tilde{r}_j,~1\leq j \leq n-1
\end{align}
where, edge $e_i \in \mathcal{E}_c$ and edge $e_j \in \mathcal{E}_\tau$ are sibling edges and 
\begin{align}
\label{inter}
r_{i}(k)=\begin{cases}
    +1,& \text{if $e_k$ is travelled in the $+$ direction, } \\
   -1,              & \text{if $e_k$ is travelled in the $-$ direction,}\\
   0, & \text{if $e_k$ is not traversed}
\end{cases}
\end{align}
in the signed path for $e_i$, for $n-1 <i \leq m$. The following result is the same as Proposition 3.10 of \cite{MesbahiEgerstedt}.
\begin{lemma}[\cite{MesbahiEgerstedt}]
For a weighted digraph $\mathcal{G}$ having a rooted in-branching and positive weights on all edges, the eigenvalues of the graph Laplacian $\bar{L}_g$ belong to the union of the open right half plane with the origin.
\end{lemma}

\begin{lemma}
The edge Laplacian $\bar{L}_e$ and the graph Laplacian $\bar{L}_g$ for a weighted directed graph (with positive weights) $\mathcal{G}$ have the same non-zero eigenvalues. 
\end{lemma}
\begin{proof} Consider a non-zero eigenvalue of the graph Laplacian $\bar{L}_g$, say $\lambda$. Hence, $\bar{L}_g x= \mathcal{A}WE^T x=\lambda x$, where $x\neq \mathbf{1}_n$. Premultiplying both sides by $E^T$, it follows that $\bar{L}_e E^T x=\lambda E^T x \neq \mathbf{0}$. Thus, $\lambda$ is also a nonzero eigenvalue of $\bar{L}_e$. Conversely, let $\mu$ be a non-zero eigenvalue of $\bar{L}_e$. Thus, $\bar{L}_e y=\mu y$. Premultiplying both sides by $\mathcal{A}W$ (clearly, $\mathcal{A}Wy \neq \mathbf{0}_n$ as $\mu\neq 0$), it is easy to see that $\bar{L}_g(\mathcal{A}W)y=\mu (\mathcal{A}W)y$. Thus, $\mu$ is also a non-zero eigenvalue of $\bar{L}_g$.
\end{proof}
\begin{lemma}In a weighted digraph with positive weights containing a rooted in-branching, the algebraic multiplicity and geometric multiplicity of the zero eigenvalue of $\bar{L}_e$ are equal to $m-n+1$.
\end{lemma}
\begin{proof} From Lemma 6, it is clear that the non-zero eigenvalues of $\bar{L}_e$ and $\bar{L}_g$ are identical. Thus, $\mathbf{rank}(\bar{L}_g)=\mathbf{rank}(\bar{L}_e)$. Since the graph $\mathcal{G}$ has rooted in-branching, the corresponding unoriented version of this branching is a spanning tree. Thus, by \cite{RenB}, $\mathbf{rank}(\bar{L}_g)=n-1$. So for $\bar{L}_e$, the algebraic multiplicity of the zero eigenvalue is $m-n+1$.
\par{}{\it Case 1}: Suppose there is only one globally reachable node in $\mathcal{G}$. Then, from Lemma 1, $\mathcal{N}(\bar{L}_e)=\mathcal{N}(\mathcal{A}W)$ since $\mathbf{1_n}\not \in \mathcal{R}(\mathcal{A})$ and $\mathbf{dim}[\mathcal{N}(\bar{L}_e)]= \mathbf{dim}[\mathcal{N}(\mathcal{A}W)]=\mathbf{dim}[\mathcal{N}(\mathcal{A})]=m-n+1$.
\par{}{\it Case 2}: Consider, multiple globally reachable nodes, in which case every node of $\mathcal{G}$ has an out-degree greater than 0. Hence, by Lemma 1, $\mathcal{N}(\mathcal{A}W)\subset \mathcal{N}({L}_e)$ and $\mathbf{dim}[\mathcal{N}(\mathcal{A}W)]=m-n$. But, $\mathbf{1_n}\in \mathcal{R}(\mathcal{A})$. So, it only remains to be shown that the vector $x$, such that $\mathcal{A}Wx=\mathbf{1_n}$, when added to the basis set of $\mathcal{N}(\mathcal{A}W)$, forms an independent set.  Without loss of generality, consider a labelling so that the first $n$ edges ($e_1$ through $e_n$) emerge from the first $n$ vertices of $\mathcal{G}$. Thus the first $n$ columns from the basis for $\mathcal{R}(\mathcal{A})$ and also the standard basis for $\mathbb{R}^n$. The remaining $m-n$ columns are identical to any one of these first $n$ columns. Hence, one choice for the basis vectors of $\mathcal{N}(\mathcal{A}W)$ will comprise vectors $y_i \in \mathbb{R}^m$ of the form $[\overbrace{0~\ldots~1/w_s~0~\ldots~0}^{n}~\underbrace{0~\ldots~-1/w_r~0~\ldots~0}_{m-n}]^T$. Moreover, for the chosen labelling, it is clear that $\mathcal{A}Wx=\mathbf{1_n}$ for $x=[1/w_1~\ldots~1/w_s~1/w_{s+1}~\ldots~1/w_n~\mathbf{0}_{m-n}^T]^T$. Observe that $x$, when augmented with $\{y_i\}_{i=1,\ldots,m-n}$, forms a linearly independent set and this augmented set forms the basis for $\mathcal{N}(\bar{L}_e)$. Thus, $\mathbf{dim}[\mathcal{N}(\bar{L}_e)]=m-n+1$ because geometric multiplicity $\leq$ algebraic multiplicity.
\end{proof}
\begin{lemma}
In a weighted digraph with positive weights and rooted in-branching, the graph Laplacian $\bar{L}_g$ is similar to 
$$\left[\begin{array}{cc}
E(\mathcal{G_\tau})^{T}\mathcal{A}(\mathcal{G})WR^{T} & \mathbf{0_{n-1}}\\
\mathbf{1_n}^T\mathcal{A}(\mathcal{G})WR^{T} & 0 \end{array} \right].$$
\end{lemma}
\begin{proof} Consider the matrices $S^{-1}=[E(\mathcal{G_\tau})~\mathbf{1_n}]^T$ and $S=[E(\mathcal{G_\tau})(E(\mathcal{G_\tau})^{T}E(\mathcal{G_\tau}))^{-1}~ \frac{1}{n}\mathbf{1_n}]$. Now $S^{-1}\bar{L}_g S=\left[\begin{array}{cc}
E(\mathcal{G_\tau})^{T}\mathcal{A}(\mathcal{G})WR^{T} & \mathbf{0_{n-1}}\\
\mathbf{1_n}^T\mathcal{A}(\mathcal{G})WR^{T} & 0 \end{array} \right]$, using \eqref{inci}-\eqref{cite}.
\end{proof}
\begin{cor} If the digraph in Lemma 8 had exactly one globally reachable node then the factorisation in \eqref{inc}-\eqref{inci1} would hold and the graph Laplacian $\bar{L}_g$ is similar to 
$$\left[\begin{array}{cc}
E(\mathcal{G_\tau})^{T}\mathcal{A}(\mathcal{G_\tau})\tilde{R}WR^{T} & \mathbf{0_{n-1}}\\
\mathbf{1_n}^T\mathcal{A}(\mathcal{G_\tau})\tilde{R}WR^{T} & 0 \end{array} \right].$$
\end{cor}
\begin{lemma}
In a weighted digraph with positive weights and rooted in-branching, the edge Laplacian $\bar{L}_e$ is similar to 
$$\left[\begin{array}{cc}
E(\mathcal{G_\tau})^{T}\mathcal{A}(\mathcal{G})WR^{T} & E(\mathcal{G_\tau})^{T}\mathcal{A}(\mathcal{G})WN_{\tau}\\
0_{(m-n+1)\times (n-1)} & 0_{(m-n+1)\times (m-n+1)} \end{array} \right],$$
where, the columns of the matrix $N_{\tau}\in \mathbb{R}^{m \times (m-n+1)}$ form the orthonormal basis for $\mathcal{N}(R)$.
\end{lemma}
\begin{proof} The matrix $R^{T} \in \mathbb{R}^{m \times (n-1)}$ has full column rank and so the left inverse $(RR^T)^{-1}R$ exists. Consider $V^{-1}=[\left((RR^T)^{-1}R\right)^T~ N_{\tau}]^T$ and $V=[R^T~N_{\tau}]$. Now, $V^{-1}\bar{L}_{e}V=$ $\left[\begin{array}{cc}
E(\mathcal{G_\tau})^{T}\mathcal{A}(\mathcal{G})WR^{T} & E(\mathcal{G_\tau})^{T}\mathcal{A}(\mathcal{G})WN_{\tau}\\
0_{(m-n+1)\times (n-1)} & 0_{(m-n+1)\times (m-n+1)} \end{array} \right]$. 
\end{proof}
\begin{cor}If the digraph in Lemma 9 had exactly one globally reachable node then the factorisation in \eqref{inc}-\eqref{inci1} would hold and the edge Laplacian $\bar{L}_e$ would be similar to 
$$\left[\begin{array}{cc}
E(\mathcal{G_\tau})^{T}\mathcal{A}(\mathcal{G_\tau})\tilde{R}WR^{T} & E(\mathcal{G_\tau})^{T}\mathcal{A}(\mathcal{G_\tau})\tilde{R}WN_{\tau}\\
0_{(m-n+1)\times (n-1)} & 0_{(m-n+1)\times (m-n+1)} \end{array} \right].$$
\end{cor}

From Lemma 7 and Corollaries 1-2, the matrix $E(\mathcal{G_\tau})^{T}\mathcal{A}(\mathcal{G_\tau})\tilde{R}WR^{T}$, for $\mathcal{G}$ (with positive weights), is invertible if it has a rooted in-branching with exactly one globally reachable node. Similarly, for multiple globally reachable nodes, $E(\mathcal{G_\tau})^{T}\mathcal{A}(\mathcal{G})WR^{T}$ is invertible. Furthermore, from Lemmas 5-9, the eigenvalues of both these matrices are in the open right half plane.
\section{Robust stability of Uncertain Directed Consensus}\label{sec:robust}

Consensus dynamics over a weighted digraph is driven by 
\begin{align}
\label{con}
\dot{x}=-\bar{L}_g x,
\end{align}
where, $x \in \mathbb{R}^{n}$ denotes the node states. Pre-multiplying both sides by $E(\mathcal{G})^T$, yields $\dot{x}_e=-\bar{L}_e x_e$ where, $x_e=E(\mathcal{G})^T x=R^T E(\mathcal{G}_{\tau})^T x\in \mathbb{R}^{m}$ denotes the edge states. Choosing a suitable transformation $z=V^{-1}x_e$, it turns out that $z=[\left((RR^T)^{-1}R\right)^T~ N_{\tau}]^T R^T E(\mathcal{G}_{\tau})^T x=[x^{T}E(\mathcal{G}_{\tau})~ \mathbf{0}_{m-n+1}^T]^T$. Thus, the first $n-1$ components of $z$ represent the edge states of the rooted in-branching.
Lemma 9 suggests that it is sufficient to concentrate on the dynamics of the edges in the rooted in-branching, say $x_{\tau}$, given by 
\begin{align}
\label{edge}
\dot{x}_{\tau}= -E(\mathcal{G_\tau})^{T}\mathcal{A}(\mathcal{G})WR^{T}x_{\tau}.
\end{align}
\par{}The notion of uncertainty is now introduced through the edge weights. The perturbations are real and are bounded about some nominal positive value. 
For this work, only additive uncertainty on a single edge weight is considered and so the weight on one of the $m$ edges is perturbed. 
This uncertainty on any edge weight $w_i$, expressed as $\delta_i$, is given by $|\delta_{i}|<\bar{\delta}, \forall i$. The uncertainty set is thus 
\begin{align}
\label{unc}
\mathbf{\Delta}=\{\Delta:\Delta=\bf \delta_i, |\delta_i| \leq\bar{\delta}< \infty\}.
\end{align}
The uncertain edge agreement protocol is
\begin{align}
\label{diag}
\dot{x}_{\tau}= -E(\mathcal{G_\tau})^{T}\mathcal{A}(\mathcal{G})(W+P_i\Delta P_i^{T})R^{T}x_{\tau},
\end{align}
with the uncertainties belonging to the set given by \eqref{unc} and $P_i\in \mathbb{R}^{m}$ is the $i$-th standard basis in $\mathbb{R}^{m}$ if the weight on edge $e_i$ is considered uncertain.
\subsection{Nyquist Stability Analysis} \label{sec:Nyq}
\begin{figure}[t]
     \centering
      \includegraphics[scale=0.3740]{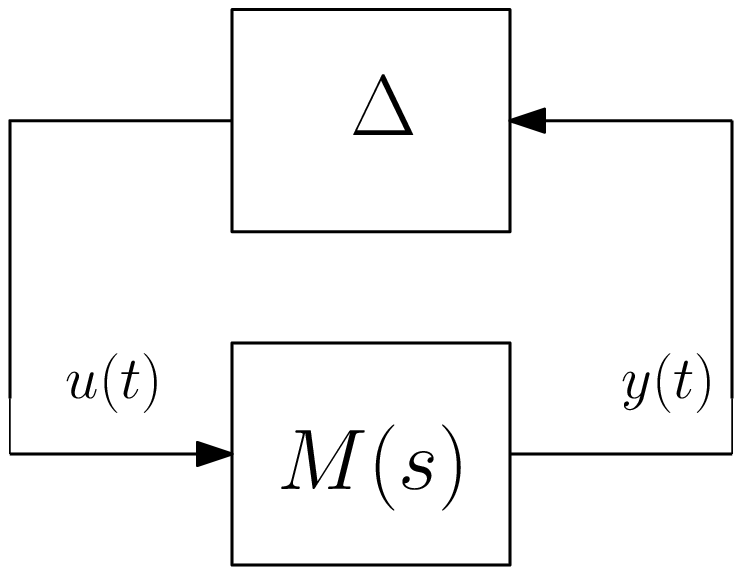}
      \caption{Uncertain consensus protocol}
      \label{UC}
      \vspace*{-0.5cm}
\end{figure}

The uncertain system, described by \eqref{diag}, is transformed in such a way that the uncertainty is separated from the nominal plant as illustrated in Fig. \ref{UC}. This formulation lends itself to a stability analysis using the Nyquist criterion. 
Consider $u$ and $y$ as the input and output, respectively, of the plant while 
the overall system is described by 
\begin{align}
\dot{x}_{\tau}&= -E(\mathcal{G_\tau})^{T}\mathcal{A}(\mathcal{G})WR^{T}x_{\tau}-E(\mathcal{G_\tau})^{T}\mathcal{A}(\mathcal{G})P_i u \label{add}\\
y&= P_i^{T}R^{T}x_{\tau},~~
u=\Delta P_i^{T}R^{T}x_{\tau}.
\end{align}
The transfer function, $M(s)$, between $y(s)$ and $u(s)$ is:
\begin{align}
\label{sys}
\hspace{-0.32cm}M(s)\hspace{-0.1cm}=-P_i^{T}R^{T}[sI+E(\mathcal{G_\tau})^{T}\mathcal{A}WR^{T}]^{-1} E(\mathcal{G_\tau})^{T}\mathcal{A}P_i.
\end{align}
The single-input single-output transfer function $M(s)$ does not have any pole at the origin because the system matrix in \eqref{add} is of full rank. The scalar uncertainty $\Delta$ can be analysed using a classical Nyquist based approach. 
\begin{theorem} The consensus protocol, \eqref{con}, over a weighted digraph $\mathcal{G}$ (with positive weights) having a rooted in-branching, is robustly stable to all perturbations $\delta_i$ on a single edge weight $w_i$, satisfying
\begin{align}
\label{rem}
|\delta_i|<GM[M(s)],
\end{align}
where $GM$ denotes the gain margin for a transfer function.
\end{theorem}
\begin{proof}
Since the transfer function $M(j\omega)$ in \eqref{sys}, as depicted in Fig. \ref{UC}, has no pole at the origin, the gain margin is obtained by computing \eqref{sys} at $s=j\omega_{pc}$ (which is the phase crossover frequency). 
Now, from the Nyquist criterion, stability dictates that $|\delta_i|<1/|M(j\omega_{pc})| $.
\end{proof}

Two special digraphs are considered next: the directed acyclic graph, having one globally reachable node, and a directed cycle graph where every node is globally reachable. 

\begin{cor}
If the digraph in Theorem 1 is acyclic, the factorization in \eqref{inc}-\eqref{inci1} holds and the limit on the perturbation on an edge, $e_i$, is given by:
\begin{align}
\label{theo}
|\delta_i|<|\left(P_i^{T}R^{T}(\tilde{R}WR^{T})^{-1}\tilde{R}P_i\right)^{-1}|
\end{align}
\end{cor}
Corollary 3 can be proved by applying to \eqref{rem}, the following 
$$(E(\mathcal{G_\tau})^{T}\mathcal{A}(\mathcal{G_\tau})\tilde{R}WR^{T})^{-1}=(\tilde{R}WR^{T})^{-1}(E(\mathcal{G_\tau})^{T}\mathcal{A}(\mathcal{G_\tau}))^{-1}$$  with $\mathcal{A}=\mathcal{A}(\mathcal{G}_\tau)\tilde{R}$ as in \eqref{inc}-\eqref{inci1}. 
\subsection{Consensus over Uncertain Directed Acyclic Graphs}

For directed acyclic graphs with a rooted in-branching, \eqref{rem} has a significant graph theoretic interpretation. The factorisations of $E$ and $\mathcal{A}$, and the subsequent interpretations of the columns of $R$ and $\tilde{R}$ 
presented in Section \ref{sec:fact}, establish this connection. 
The following results lead to such an interpretation of \eqref{rem} for directed acyclic graphs. 
\begin{lemma}
For a directed acyclic graph $\mathcal{G}$, if $\tilde{r}_i= \tilde{r}_j= q_j,~1\leq j \leq n-1$, then $r_{i}(j)=+1$, where $q_j$ is the $j$-th standard basis for $\mathbb{R}^{n-1}$.
\end{lemma}
\begin{proof}Clearly $\tilde{r}_j= q_j,$ for $1\leq j \leq n-1$ follows from \eqref{inci1}. Now, $\tilde{r}_i= \tilde{r}_j$ implies that $i\geq n$ and $e_i \in \mathcal{E}_c$. Since there are no directed cycles in $\mathcal{G}$, so any edge $e_i$ emerging from a node, say node $p$, cannot terminate at a node $t$ such that there is a directed path from node $t$ to $p$. Hence, the equivalent signed path, in $\mathcal{G}_\tau$, corresponding to the edge $e_i$, must traverse its sibling edge $e_j$ in the positive sense. Thus, $r_{i}(j)=+1$. 
\end{proof}
\begin{lemma}
For the directed acyclic graph $\mathcal{G}$ having two edges $e_s$ and $e_t$ in $\mathcal{E}_{c}$ that are siblings to edges $e_p,~e_q \in \mathcal{E}_{\tau}$, respectively, if the signed path, in $\mathcal{G}_{\tau}$, of $e_s$ includes $e_q$, then the signed path, in $\mathcal{G}_{\tau}$, of $e_t$ cannot include the edge $e_p$.
\end{lemma}
\begin{proof} It suffices to prove that $r_{s}(q)=\pm1$ implies $r_{t}(p)=0$. From Lemma 10, $r_{s}(p)=+1$. Suppose $r_{s}(q)=+1$. This means that there is a directed path through $e_p$ to the globally reachable node with $e_q$ appearing after edge $e_p$ in the sequence. Hence, any edge that is a sibling of $e_q$ (such as $e_t$) in $\mathcal{G}$ cannot be represented by a signed path that contains edge $e_p$ as this will imply the existence of a directed cycle. So $r_{t}(p)=0$.
Next, consider $r_{s}(q)=-1$. This means that a sibling edge of $e_p$ is encoded by a path that involves traversing $e_q$ in the opposite sense. Clearly, the directed path through $e_p$ to the globally reachable node does not include the edge $e_q$ and vice versa. Thus, any sibling edge of $e_q$ cannot be represented by a signed path that involves traversing $e_p$ in the positive sense either. So, $r_{t}(p)$ cannot equal $+1$. Suppose $r_{t}(p)=-1$. But this means that there is a directed path through $e_s$ and $e_t$, back to the parent node of $e_p$ and $e_s$, thereby completing a directed cycle. Thus, $r_{t}(p)=0$ is the only possibility. 
\end{proof}

Eqn. \eqref{rem} suggests that an interpretation of the perturbation bound involves an investigation of the structure of $[\tilde{R}WR^{T}]^{-1}$. Consider the matrix $\tilde{R}WR^{T}=W_{\tau}+\tilde{T}_{\tau}W_{c}T_{\tau}^T$ (using \eqref{inci} and \eqref{inc}) where $W_{\tau} \in \mathbb{R}^{(n-1)\times (n-1)}$ and $W_{c} \in \mathbb{R}^{(m-n+1)\times (m-n+1)}$ are diagonal matrices containing the weights of the edges in $\mathcal{E}_{\tau}$ and $\mathcal{E}_{c}$, respectively. From \eqref{inci}-\eqref{inci1}, the columns of $T_{\tau}$ and $\tilde{T}_{\tau}$ are the columns $n$ through $m$ of $R$ and $\tilde{R}$, respectively. Thus, $\tilde{R}WR^{T}= W_{\tau} +\sum_{i=n}^{m}w_{i}\tilde{r}_{i}r_{i}^{T}$. Now, using the Sherman-Morrison formula for inverse of rank one updates \cite{CM} iteratively, $D_{m-n+2}=(\tilde{R}WR^{T})^{-1}$ can be obtained as edges in $\mathcal{E}_c$ are added one by one to the rooted in-branching, $\mathcal{G}_\tau$, with the initial value $D_1= W_{\tau}^{-1}$ and the update rule given by
\begin{align}
\label{update}
D_{i+1}=D_{i}-\frac{w_{n+i-1}D_{i}\tilde{r}_{n+i}r_{n+i-1}^{T}D_{i}}{1+w_{n+i-1}r_{n+i-1}^{T}D_{i}\tilde{r}_{n+i-1}}.
\end{align}
It follows from \eqref{update} that for each additional edge $e_k \in \mathcal{E}_c$ incorporated, the $j$-th row, corresponding to its sibling edge $e_j \in \mathcal{E}_\tau$, is updated. Moreover, only those entries of the $j$-th row which correspond to edges in $\mathcal{G}_\tau$ that comprise the equivalent signed path of $e_k$ are updated. For instance, in Fig. \ref{demo}, when $e_8$ is added, only $[D_{i}]_{22}$, $[D_{i}]_{23}$, $[D_{i}]_{25}$ and $[D_{i}]_{26}$ in the second row will be updated. Only rows that have already been updated at earlier iterations can be affected. 
\begin{theorem} The consensus protocol, over a weighted directed acyclic graph  $\mathcal{G}$, with positive weights and a rooted in-branching, is robustly stable to all perturbations $\delta_i$ on edge weight $w_i$, if the sum of the out-degree weights of the parent node of edge $e_i$ is positive.
\end{theorem} 
\begin{proof} Suppose the robustness of edge $e_k$ is to be investigated. If $e_k$ has no sibling edges, then it is a part of $\mathcal{G}_\tau$. Also, a zero or a negative weight on $e_k$ will cut off the parent node of $e_k$, so that its state will not change or diverge from those of other nodes. Thus, the edge weight for $e_k$, which is the sum of out degree weights of its parent node, must be positive.  
Suppose $e_k$ has one or more sibling edges. Call this set of siblings of $e_k$ (including $e_k$) as $S_k$ with $|S_k|=u>0$. Choose a rooted in-branching $\mathcal{G}_\tau$ for $\mathcal{G}$. Clearly, either $e_k$ or one of its siblings is part of $\mathcal{G}_\tau$. Now, consider the set $\mathcal{E}_k=\mathcal{E}_c \setminus S_k$. None of the edges in $\mathcal{E}_k$ is a sibling to $e_k$. Suppose edges in $\mathcal{E}_k$ are added to $\mathcal{G}_\tau$ one by one. The matrix $D_2$, after addition of $e_n$, is then given by
\begin{align}
D_2=\begin{bmatrix}1/w_1 & \ldots & 0 & \ldots & 0 \\
\vdots & \ddots & \vdots & \vdots & \vdots \\
 l_1 & \ldots & \frac{1}{w_j+w_n}& \ldots &  l_{n-1} \\
\vdots & \vdots & \vdots & \ddots & \vdots \\
0 & \ldots & 0 & \ldots & 1/w_{n-1}
\end{bmatrix}
\label{ed1}
\end{align}
where, $e_n \in \mathcal{E}_k$ is a sibling to $e_j \in \mathcal{E}_\tau$.
and $l_i=\mp \frac{ w_n/w_i w_j}{1+w_n/w_j}$ or $0$, depending on whether edge $e_i \in \mathcal{E}_\tau$ forms a part of the signed path vector for $e_n$ 
or not. Similarly, if $e_{n+1}\in \mathcal{E}_k$ is also a sibling of $e_j$, only the $j$-th row of $D_2$ is similarly updated in $D_3$ with $[D_3]_{jj}=\frac{1}{w_j+w_n+w_{n+1}}$. However, if $e_{n+1}$ is a sibling of $e_s \in \mathcal{E}_\tau$ ($s \neq j$) then only the $s$-th row and possibly the $j$-th row is updated from $D_2$ to $D_3$. 
This is because $[w_{n+1}D_{2}\tilde{r}_{n+1}]$ in \eqref{update}, which selects the $s$-th column of $w_{n+1}D_{2}$, has a non-zero component at the $s$-th position and possibly the $j$-th position. Extending this argument by induction, it follows that while updating from $D_q$ to $D_{q+1}$, only those rows corresponding to edges in $\mathcal{G}_\tau$ whose siblings have been added in the first $q$ steps are updated. Suppose all edges in $\mathcal{E}_k$ are added so that only edges in $S_k$ remain to be added in order to obtain $(RW\tilde{R}^{T})^{-1}$. Clearly, till this step, the $k$-th row has not been updated. When the first edge $e_{k1}$ in $S_k$ is added, the $k$-th row is updated for the first time with the $k$-th diagonal entry changing from $1/w_k$ to $\frac{1}{w_k+w_{k1}}$. This follows from Lemmas 10-11 as $[w_{k1}D_{m-n-u+2}\tilde{r}_{k1}]$ has non-zero component at $k$-th position and at those positions which correspond to edges in $\mathcal{G}_\tau$ whose siblings' signed path included $e_k$ 
Thus, other than at the $k$-th position (which has an entry of $+1$), the vector $r_{k1}$ will have a $0$ at exactly those positions where $[w_{k1}D_{m-n-u+2}\tilde{r}_{k1}]$ has non-zero entries. Hence, the $k$-th component of $r_{k1}D_{m-n-u+2}$ is simply $1/w_k$. Also, $1+w_{k1}r_{k1}^{T}D_{m-n-u+2}\tilde{r}_{k1}=1+w_{k1}/w_{k}$. Subsequently, as other edges in $S_k$ are added till it is exhausted, the $k$-th diagonal entry, at any stage, contains the reciprocal of the out-degree sum of the parent node of $e_k$. 

The premultiplication and postmultiplication of $D_{m-n+2}=(RW\tilde{R})^{-1}$ by $P_i^{T}R^{T}$ and $\tilde{R}P_i$, respectively, chooses the $k$-th diagonal entry, due to the complementary nature of the non-zero entries of $r_{k}$ and $D_{m-n+2}\tilde{r}_{k1}$ except at the $k$-th position. Since $e_k$ was chosen arbitrarily, the margin for perturbation of an edge weight is equal to the sum of the out-degrees of the parent node. 
\end{proof} 
\subsection{Consensus over Uncertain Cycle Digraph}

Theorem 2 deals with a digraph having exactly one globally reachable node. In the cycle digraph however, all the $n$ nodes are globally reachable. Removing any one of the edges from a cycle digraph results in the rooted in-branching. Since the cycle graph has multiple globally reachable nodes, the relation in \eqref{inc}-\eqref{inci1} does not hold. 
But a suitable similarity transformation of the edge and graph Laplacians leads to a block diagonal matrix in this case, instead of block triangular ones, and for the cycle digraph $\mathcal{A}=I_n$. The cycle digraph is specially important as it lies at the heart of the well known cyclic pursuit algorithm \cite{klamkin, gagnon, bruck, marshall, beard, AS, me1, me2}. Some relevant results are now stated.
\begin{lemma}The graph Laplacian for weighted cyclic pursuit, $\bar{L}_g=\mathcal{A}WE(\mathcal{G})^T$ is similar to $\left[\begin{array}{cc}
E(\mathcal{G_\tau})^{T}WR^{T} & 0\\
0 & 0 \end{array} \right]$.
\end{lemma}
\begin{proof} Consider the matrices $$S_{1}=\left[WE(\mathcal{G_\tau})(E(\mathcal{G_\tau})^{T}WE(\mathcal{G_\tau}))^{-1} \quad\mathbf{1}_n\left(\sum_i\frac{1}{w_i}\right)^{-1}\right]$$ and $S_{1}^{-1}=\left[\begin{array}{c} E(\mathcal{G_\tau})^{T}\\ \mathbf{1}_n^{T}W^{-1}\end{array}\right]$. It is at once apparent that $S_{1}^{-1}\bar{L}_{g}S_{1}=\left[\begin{array}{cc}
E(\mathcal{G_\tau})^{T}WR^{T} & 0\\
0 & 0 \end{array} \right]$. 
\end{proof}
\begin{lemma}
The edge Laplacian for weighted cyclic pursuit, $\bar{L}_e=E(\mathcal{G})^{T}\mathcal{A}W$ is similar to $\left[\begin{array}{cc}
E(\mathcal{G_\tau})^{T}WR^{T} & 0\\
0 & 0 \end{array} \right]$.
\end{lemma}
\begin{proof}Consider the matrices $S_{2}=[R^{T}~~~W^{-1}\mathbf{1}_n]$ and $S_{2}^{-1}=\left[\begin{array}{c} (R(\mathcal{G_\tau})WR(\mathcal{G_\tau})^{T})^{-1}R(\mathcal{G_\tau})W\\\left(1/\sum_i\frac{1}{w_i}\right)\mathbf{1}_n^{T}\end{array}\right]$. It follows that $S_{2}^{-1}\bar{L}_{e}S_{2}=\left[\begin{array}{cc}
E(\mathcal{G_\tau})^{T}WR^{T} & 0\\
0 & 0 \end{array} \right]$. 
\end{proof}
\begin{lemma}
For the weighted cycle digraph, the edge Laplacian is similar to the graph Laplacian.
\end{lemma}
\begin{proof} Since both the graph Laplacian, $\bar{L}_g$ and the edge Laplacian, $\bar{L}_e$ for the cyclic digraph are similar to the matrix $\left[\begin{array}{cc}
E(\mathcal{G_\tau})^{T}WR^{T} & 0\\
0 & 0 \end{array} \right]$, so using the transformation $S^{-1}\bar{L}_{e}S$, with $S=S_{2}S_{1}^{-1}$, the result follows. 
\end{proof}
Thus, the reduced edge version of cyclic pursuit is 
\begin{align}
\label{red1}
\dot{x}_{\tau}=-E(\mathcal{G_\tau})^{T}WR^{T}x_{\tau}.
\end{align}
Considering a perturbation in $w_1$, it follows that
\begin{align}
\label{diagcp}
\dot{x}_{\tau}=-E(\mathcal{G_\tau})^{T}(W+P_i\Delta P_i^{T})R^{T}x_{\tau},
\end{align}
with the uncertainties belonging to the set given by \eqref{unc} and $P_i\in \mathbb{R}^{n}$ is a $\{0,1\}$ vector with 0-entries everywhere except at $[P]_{1}$. This is because in the cycle graph every edge is equivalent and without loss of generality the perturbation may be considered in $w_1$. {Here too, the phase crossover occurs at $\omega=0$} and so $M(0)$ is explicitly computed to be $M(0)=-\dfrac{\sum_{i=2}^{n}\frac{1}{w_i}}{1+w_1\sum_{i=2}^{n}\frac{1}{w_i}}$ \cite{DM1}. The Nyquist criteria yields
\begin{align}
\label{fr}
-w_1-\dfrac{1}{\sum_{i=2}^n\frac{1}{w_i}}<\bar{\delta}\Rightarrow w_{1} + \bar{\delta}>-\dfrac{1}{\sum_{i=2}^n\frac{1}{w_i}}.
\end{align}
Thus, the robust stability criterion for cyclic pursuit is stated in the following theorem, similar to \cite{AS}. 
\begin{figure}[t]
     \centering
      \includegraphics[scale=0.50]{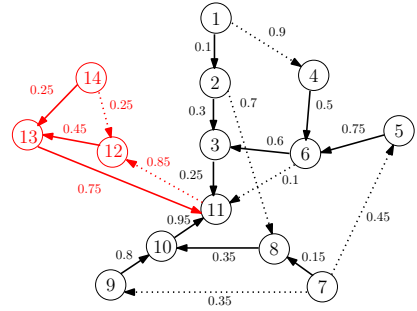}
      \caption{Weighted digraph in the examples (black portions for first example, black+red for second example).} 
      \label{WD}
      \vspace*{-0.5cm}
\end{figure}

\begin{figure*}[!ht]
     \centering
     \begin{tabular}{ccc}
     \hspace*{-0.55 cm}
     \subfloat[]{
      \includegraphics[scale=0.35]{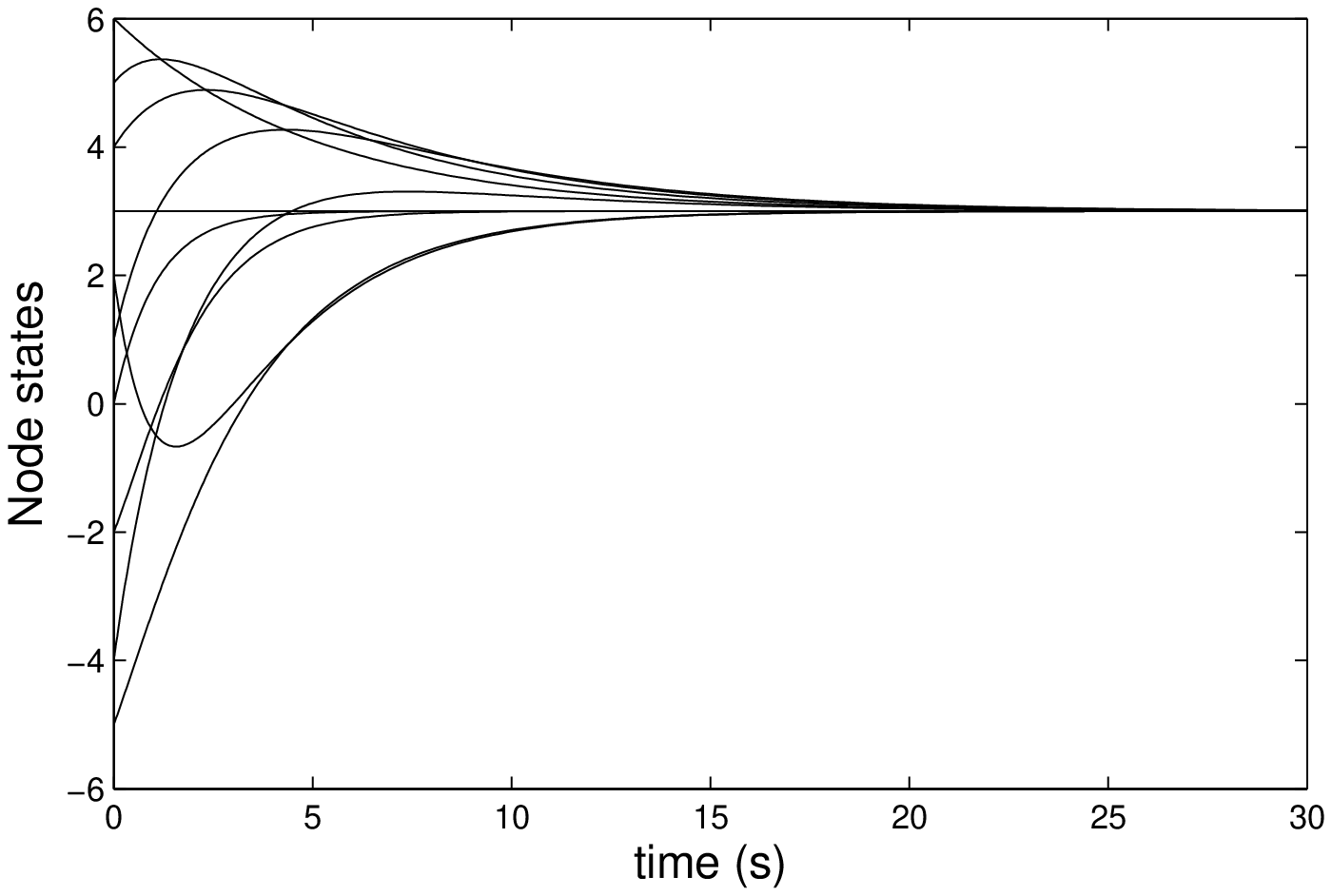}
      \label{g12}}
      \hspace*{-0.335 cm}
        \vspace*{-0.120 cm}
    \subfloat[]{
      \includegraphics[scale=0.405]{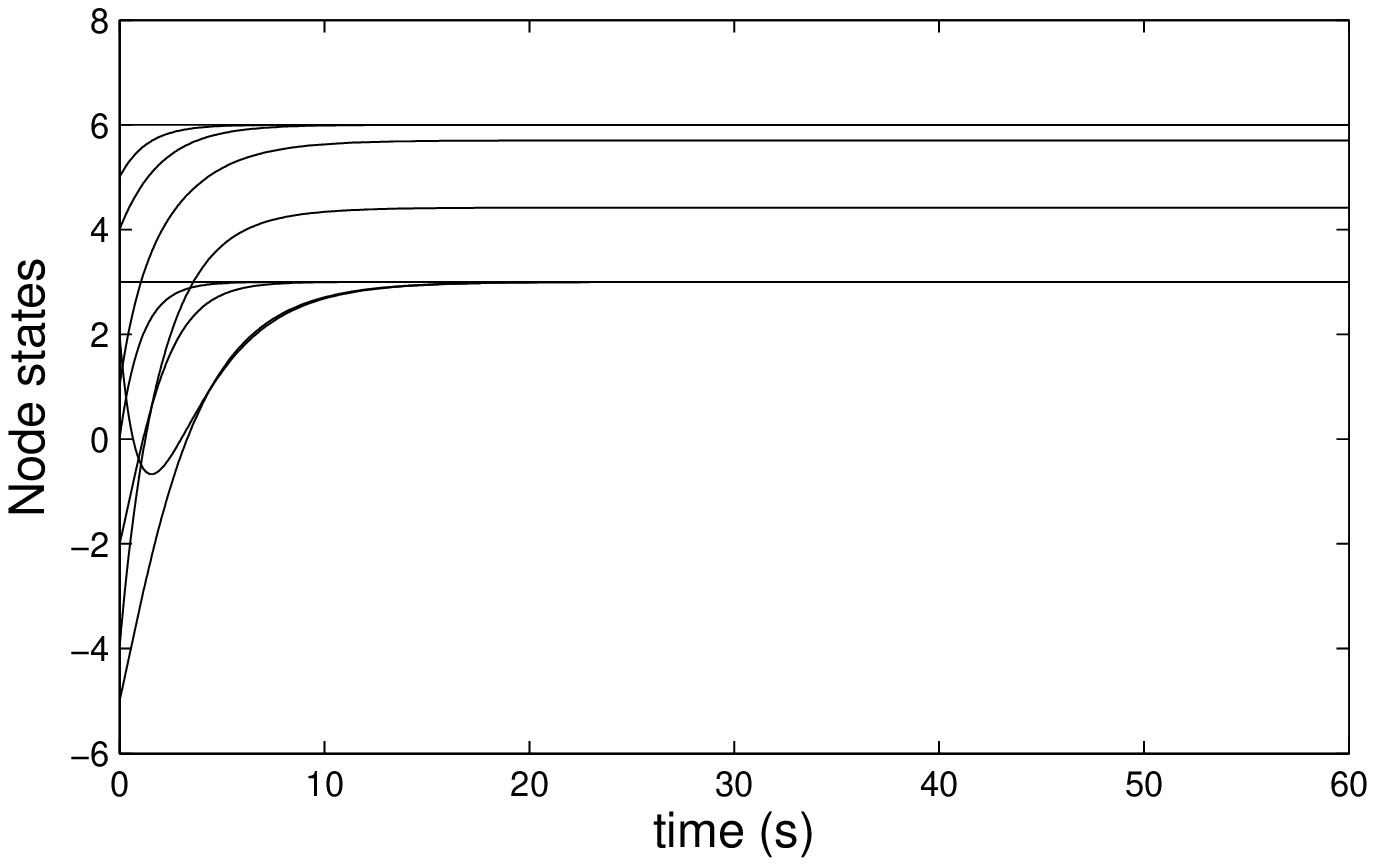}
      \label{g23}}
      \hspace*{-0.335 cm}
        \vspace*{-0.120 cm}
     \subfloat[]{
      \includegraphics[scale=0.40]{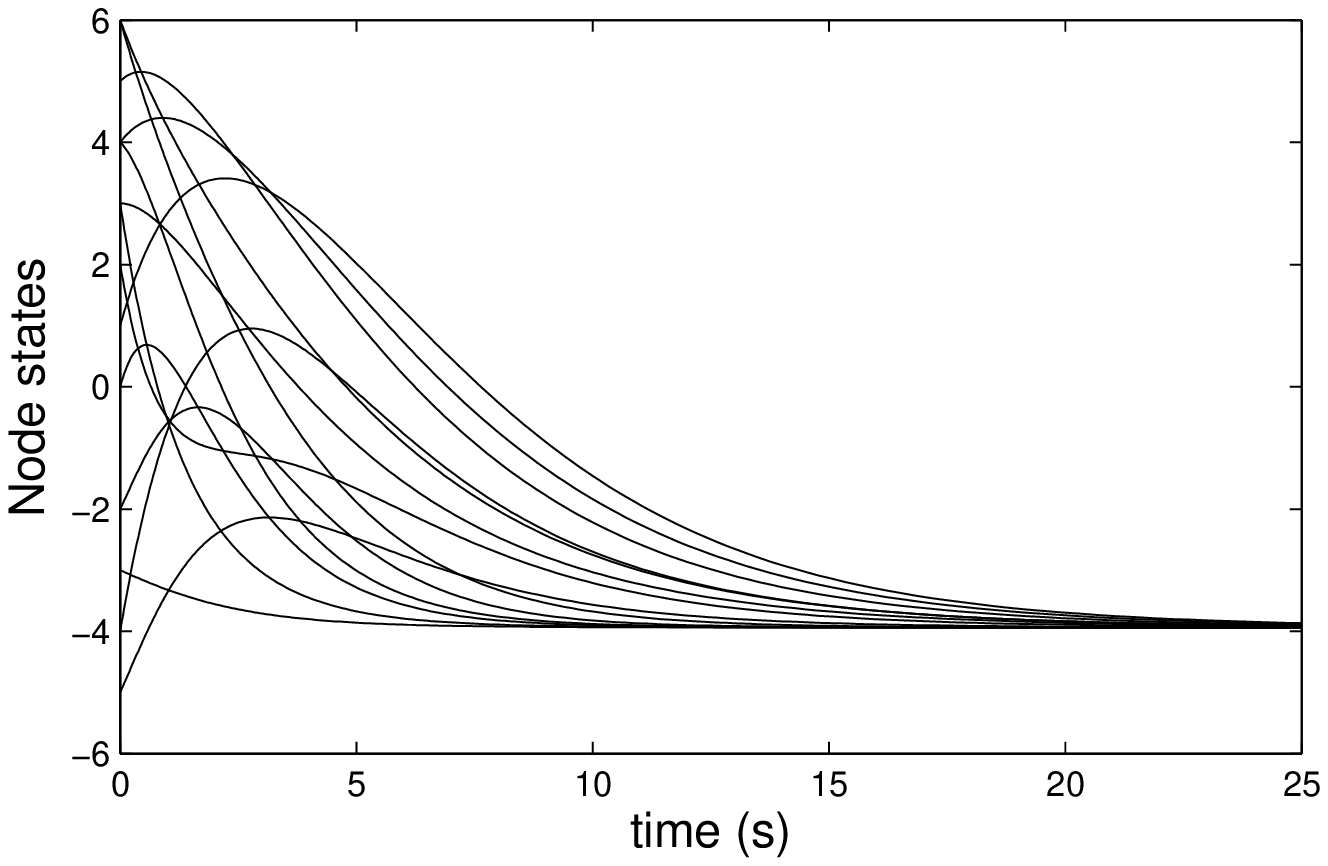}
      \label{f12}}
      \hspace*{-0.4750 cm}
      \end{tabular}
      \caption{Node states for perturbed weight on $e_{6,11}$ (a) within tolerable bound, (b) at exact bound, in first example and (c) for perturbed weight on $e_{12,13}$ in second example within bound.}
        \vspace*{-0.50 cm}
\end{figure*}

\begin{theorem}
Given a perturbation on a single edge, say $e_j$ (with nominal weight $w_j$), the heterogeneous cyclic pursuit system is stable for perturbations bounded below by $\bar{\delta}$:
 \begin{align}
\label{cond}
\bar{\delta}>-w_j-\dfrac{1}{\sum_{i=1, i\neq j}^n\frac{1}{w_i}}.
\end{align}
\end{theorem}
For the cycle graph, the limit on $w_j+\bar{\delta}$ is the equivalent resistance between the vertices $j$ and $j+1$ when the edge, $e_j$, joining nodes $j$ and $j+1$, is removed. The reciprocal of the edge weight is the resistance corresponding to each edge. In \cite{DZ1}, it was shown that for consensus over an undirected graph, an edge weight can be negative so long as this negative value is greater than a bound that equals the negative of the equivalent resistance between the vertices that the perturbed edge joins. This same interpretation holds for the directed cycle graph.

%

\section{Simulation Results}
\label{sec: sim}
Consider the weighted directed acyclic graph $\mathcal{G}$, in Fig. \ref{WD} (black portions only), with 11 nodes and 15 edges. The bold edges denote a rooted in-branching with a single globally reachable node 11. The dotted edges belong to $\mathcal{E}_{c}$. The nominal positive edge weights are shown. The edge, $e_{6,11}$ is assumed to be perturbed. 
The initial node states are $[1~ 2~ 3~ 4~ 5~ 6~ -4~ -5~ -2~ 0~ 3]$. In Fig. \ref{g12}, the perturbation on the edge weight is $-0.50$, so the perturbed weight is $-0.40$. It may be seen that consensus is achieved. In Fig. \ref{g23}, where the perturbation is exactly equal to the bound, that is $-0.70$ (computed from \eqref{rem}), so that the perturbed weight is $-0.60$ the nodes form clusters. With a perturbation of $-1.00$, consensus is not achieved as the node states diverge in this case. The Nyquist plots for convergent, clustering and divergent cases are shown in Fig. \ref{Ny} with the black dot representing the critical point $(-1,0)$. 
\par{}Next, in the graph in Fig. \ref{WD} with both the black and red portions (14 nodes and 20 edges), nodes 11, 12 and 13 are globally reachable. Hence, \eqref{rem} of Theorem 1 is used to obtain perturbation limits on the edge weights. 
A perturbation of $-0.50$ is applied to the weight on $e_{12,13}$, while the critical value is $-0.85$ and the corresponding convergent evolution of the node states is shown in Fig. \ref{f12}. 


\section{Conclusions}
\label{sec: conc}

This paper presented an analysis of the robustness margins for the edge weights of a weighted directed graph having a rooted in-branching. Although only one weight is perturbed at a time, the presented framework is suitable for analysis of multiple uncertain edge weights by employing small gain theorem. However, using present results, for any directed graph, it may be determined as to which edge is the most vulnerable. In other words, if an `attacker' wants to disrupt the consensus protocol, the present set up enables one to choose the most vulnerable edge. By suitable transformations of the edge and graph Laplacians and by considering a reduced order system the stability margin of the consensus protocol can thus be determined without explicit eigenvalue computations. Graph theoretic interpretations of the robustness margins for a directed acyclic graph and a directed cycle graph provide further insights and serve as an encouragement to interpret the result for more general graphs. 
\begin{figure}[!ht]
     \centering
      \includegraphics[scale=0.475]{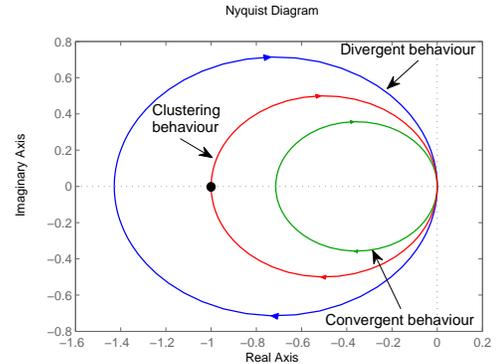}
      \caption{Nyquist plots of $M(s)\Delta$ for first example with uncertain weight on $e_{6,11}$ for the three types of behaviour.}
       \label{Ny}
       \vspace*{-0.5cm}
\end{figure}

\bibliography{bibtex_database}
\bibliographystyle{IEEEtran}

\end{document}